\documentclass[11pt,letterpaper]{amsart}

\usepackage[left=1in,top=1in,right=1in,bottom=1in]{geometry}
\usepackage[foot]{amsaddr}
\usepackage{amsmath}
\usepackage{amsthm}
\usepackage{amssymb}
\usepackage{url}
\usepackage{hyperref}
\usepackage{mathtools}

\usepackage{todonotes}

\pdfoutput=1

\newtheorem{Thm}{Theorem}
\newtheorem{Lem}[Thm]{Lemma}
\newtheorem{Cor}[Thm]{Corollary}

\theoremstyle{remark}

\theoremstyle{definition}

\theoremstyle{definition}

\begin{document}

%%%%%%%%%%%%%%%%%%%%%%%%%%%%%%%%%%%%%%%%%%%%%%%%%% Front-matter information %%%

\title[]{An Asymptotically Fast Polynomial Space Algorithm for Hamiltonicity Detection in Sparse Directed Graphs}
\author{Andreas Bj\"orklund}
\thanks{This work was carried out while working as a researcher for Ericsson Research}

\begin{abstract}
We present a polynomial space Monte Carlo algorithm that given a directed graph on $n$ vertices and average outdegree $\delta$, detects if the graph has a Hamiltonian
cycle in $2^{n-\Omega(\frac{n}{\delta})}$ time. This asymptotic scaling of the savings in the running time matches the fastest known exponential space algorithm by Bj\"orklund and Williams ICALP 2019.
By comparison, the previously best polynomial space algorithm by Kowalik and Majewski IPEC 2020 guarantees a $2^{n-\Omega(\frac{n}{2^\delta})}$ time bound.
%By comparison, the previously best polynomial space algorithm for a more restricted version where the sum of in- and outdegree at every vertex is bounded by $d$, implicit in Bj\"orklund \emph{et al.} ACM Trans. Algorithms 2012, only guarantees a $2^{n-\Omega(\frac{dn}{2^d})}$ time bound. % with $\epsilon_d=d/{2^{d}}$.

Our algorithm combines for the first time the idea of obtaining a fingerprint of the presence of a Hamiltonian cycle through an inclusion--exclusion summation over the Laplacian of the graph from Bj\"orklund, Kaski, and Koutis ICALP 2017, with the idea of sieving for the non-zero terms in an inclusion--exclusion summation by listing solutions to systems of linear equations over $\mathbb{Z}_2$ from Bj\"orklund and Husfeldt FOCS 2013.
\end{abstract}

\maketitle
\thispagestyle{empty} % caveat: must occur after maketitle to have effect

%%%%%%%%%%%%%%%%%%%%%%%%%%%%%%%%%%%%%%%%%%%%%%%%%%%%%%%%%%%%% Document body %%%

\clearpage

\section{Introduction}
Given a directed graph $G=(V,A)$ on $n=|V|$ vertices, we consider the problem of detecting if $G$ has a Hamiltonian cycle,
a directed cycle through $G$ using a subset of the arcs $A$, visiting each vertex of $V$ exactly once. We call this the \emph{Hamiltonicity} problem. Deciding Hamiltonicity in a directed graph is one of Karp's original NP-complete problems~\cite{Karp1972}. For a very long time, the best worst case algorithm known for this problem was based on Bellman's~\cite{Bellman1962} and Held and Karp's~\cite{HeldK1962} dynamic programming across all vertex subsets from the early 1960's running in $2^n\operatorname{poly}(n)$ time. If and when one can improve over $O^*(2^n)$ time has been the focus of much of recent research, confer the related work section below.
A recent result by Bj\"orklund and Williams~\cite{BjorklundW2019}, building on Bj\"orklund, Kaski, and Koutis~\cite{BjorklundKK2017}, describes a deterministic, exponential space, $2^{n-\Omega(\frac{n}{\delta})}$ time algorithm that counts the number of Hamiltonian cycles where $\delta=\frac{|A|}{n}$ is the average outdegree. A natural follow-up question is whether this speedup intrinsically comes at the cost of exponential space usage. In other words, is the obtained speedup necessarily an effect of efficient tabulation of solutions to recurrent subproblems?
%While it was discovered a decade ago by Bj\"orklund~\cite{Bjorklund2014} that Hamiltonicity in undirected graphs admits a much faster randomised $1.657^n\operatorname{poly}(n)$ time algorithm, no $c^n\operatorname{poly}(n)$ time algorithm for any $c<2$ is known for Hamiltonicity in directed graphs. Hence much research in the past decade has been done on restricted inputs where such algorithms nevertheless can be proven to exist. In particular interest to the present paper, is the recent result by Bj\"orklund and Williams~\cite{BjorklundW2019}, building on Bj\"orklund, Kaski, and Koutis~\cite{BjorklundKK2017}, on sparse graphs. They give a deterministic, exponential space, $2^{n-\Omega(\frac{n}{\delta})}$ time algorithm that counts the number of Hamiltonian cycles where $\delta=\frac{|A|}{n}$ is the average outdegree. 
In this paper, we give a partial negative answer to that question. We show that this requirement of an exponentially sized space resource can be reduced to a polynomially sized one, when we are only interesting in detecting if the graph has a Hamiltonian cycle, and are content with a randomised algorithm. 
We prove
\begin{Thm}
\label{thm:main}
There is a polynomial space Monte Carlo algorithm that given an $n$-vertex directed graph of average outdegree $\delta$, detects w.h.p. if the graph has a Hamiltonian cycle in $2^{n-\Omega(\frac{n}{\delta})}$ time, without any false positives.
\end{Thm}

Our algorithm builds on the algorithms by Bj\"orklund, Kaski, and Koutis ~\cite{BjorklundKK2017} and its successor by  Bj\"orklund and Williams~\cite{BjorklundW2019}.
At the core of the algorithms in~\cite{BjorklundKK2017} and~\cite{BjorklundW2019}, are efficient methods to list contributing terms to a sum evaluating to the number of Hamiltonian cycles. They use split, tabulate, and list procedures that seem to require exponential space. In more detail, in~\cite{BjorklundW2019}, they reduce the problem to listing pairs of dissimilar vectors from two exponential size sets of short vectors, where dissimilar means different in each coordinate. They further present two efficient algorithms to solve this latter problem that build on tabulation, one explicitly on one of the two sets, and the other indirectly as subresults of a large fast matrix multiplication. 
Our overall algorithm use a similar idea of listing contributing terms to a sum, but we use a different approach of obtaining the terms. In particular, we do not directly reduce to the problem of listing dissimilar vectors.  Our main insight is that another technique previously used to compute the parity of the number of Hamiltonian cycles by Bj\"orklund and Husfeldt~\cite{BjorklundH2013}, by a careful design, can replace the tabulation for enumeration of solutions to a linear equation system over $\mathbb{Z}_2$. This latter task is well-known to be possible to do in polynomial space. Our way of combining the above two techniques is our main technical novelty.

Polynomial space algorithms improving over $O^*(2^n)$ time in sparse graphs were known before.
We note that for the easier case of everywhere sparse graphs, i.e., graphs in which the sum of the in- and outdegree at every vertex is bounded by $d$,
Bj\"orklund \emph{et al.}~\cite{BjorklundHKK2012} implicitly showed that you can decide Hamiltonicity in $2^{n-\Omega(\frac{dn}{2^d})}$ time using polynomial space. Their paper considered TSP in undirected graphs, but it is not difficult to see that their proof of Theorem 1.3 could also be used for directed Hamiltonicity. 
Very recently, Kowalik and Majewski~\cite{KowalikM2020} presented a polynomial space, $2^{n-\Omega(\frac{n}{2^\delta})}$ time algorithm for directed Hamiltonicity on $n$-vertex graphs of average outdegree $\delta$. It builds on the
algorithm by Bj\"orklund~\cite{Bjorklund2018} which is a $(2-2^{1-\delta})^{n/2}\operatorname{poly}(n)$ time polynomial space algorithm in undirected bipartite graphs of average degree $\delta$ using techniques similar in spirit to our algorithm design here. Note though that in our design, the speedup is exponential in $n/\delta$, whereas the speedup in~Kowalik and Majewski~\cite{KowalikM2020} is exponential in $n/2^\delta$.
%As far as we know, no faster-than-$2^n$ time polynomial space algorithm was known for constant $\delta$, apart for algorithms for very low degree graphs based on branching.
%Later, Cygan and Pilipczuk~\cite{CyganP2015}  considered average degree in undirected graphs and came up with a $2^{n-\Omega(\epsilon_\delta n)}$ time algorithm for TSP in undirected graph, with $\epsilon_\delta=1/2^{2^\delta}$ using exponential space. However, it is not immediate how to use their main construction for directed Hamiltonicity in polynomial space. Indeed, all the stated results in their paper are for exponential space. Even
%if it would be possible, the dependence on $\delta$ is much worse than in the algorithm in Theorem~\ref{thm:main}.
 
 It should be noted in passing, that for many hard combinatorial problems the best known worst case algorithms use exponential space. In fact, in some cases the only known algorithms that improve over a straight-forward brute-force algorithm testing all possibilities use exponential space, are deterministic, and are also able to count the solutions. To give just one example that is also on Karp's list~\cite{Karp1972}, this holds presently for \textsc{MaxCut}: compute a bipartition of the vertices that maximises the number of edges between the two parts. It has a $O^*(1.73^n)$ time counting algorithm where $n$ is the number of vertices~\cite{Williams2004}, but no polynomial space algorithm improving substantially over the brute-force $O^*(2^n)$ time algorithm is known, even if we only consider detection and use randomisation. 

One reason to get rid of exponential space usage from a practical point of view, is to offer better implementability on parallel computing devices. This seems to be the case also with the present algorithm compared to the exponential space algorithms in~\cite{BjorklundW2019}. In particular, the computationally heavy steps in our algorithm can easily be scheduled to compute different parts of the sum that may run obliviously of each other on different processors, only adding up their final partial sums in the end.
 
\subsection{Related work}
The fastest known (exponential space) algorithm for directed Hamiltonicity as far as we know is the $2^{n-\Omega(\sqrt{n/\log\log n})}$ time algorithm by Bj\"orklund, Kaski, and Williams~\cite{BjorklundKW2019}. 
The fastest known polynomial space algorithm is the $2^n\operatorname{poly}(n)$ time one based on counting closed walks via adjacency matrix powering and inclusion--exclusion,
discovered at least four times~\cite{KohnGK1977,Karp1982,Bax1993,Barvinok1996}, the oldest by Kohn, Gottlieb, and Kohn~\cite{KohnGK1977} dates back to 1977.
Much faster algorithms exist for special cases, also apart from the $2^{n-\Omega(\frac{n}{\delta})}$ time exponential space algorithm in average outdegree $\delta$ directed graphs by Bj\"orklund and Williams~\cite{BjorklundW2019}.
In \emph{bipartite} directed graphs, there is a $O^*(1.732^n)$ time, polynomial space, algorithm by Bj\"orklund, Kaski, and Koutis~\cite{BjorklundKK2017}.
There is also an earlier $O^*(1.888^n)$ time, exponential space, algorithm by Cygan, Kratsch, and Nederlof~\cite{CyganKN2018} based on a different technique.
 Bj\"orklund and Husfeldt~\cite{BjorklundH2013} show a $O^*(1.619^n)$ time, polynomial space, algorithm that computes the \emph{parity} of the number of Hamiltonian cycles in a directed graph. Somewhat perplexingly, this algorithm does not seem to be useful for the detection problem in general. However, counting modulo powers of small primes can be used for detection when the number of Hamiltonian cycles are less than $c^n$ for some constant $c$: Bj\"orklund, Kaski, and Koutis~\cite{BjorklundKK2017}, improving over a partial result in Bj\"orklund, Dell, and Husfeldt~\cite{BjorklundDH2015}, show how to find a Hamiltonian cycle in $O((2-\epsilon_c)^n)$ time, with $\epsilon_c>0$ being another constant depending only on $c$.
In undirected graphs, there is a $O^*(1.657^n)$ time, polynomial space, algorithm, and in bipartite undirected graphs, there is a $O^*(1.415^n)$ time algorithm, both by Bj\"orklund~\cite{Bjorklund2014}. 
Despite the partial positive results above, it is a major open question in the area of exact exponential time algorithms, whether or not a $O(c^n)$ time algorithm for any $c<2$ exists for detecting Hamiltonian cycles in general directed graphs.

\subsection{Methodology}
Our algorithm is based on algebraic fingerprinting for Hamiltonian cycles, following a long line of works~\cite{Bjorklund2014,BjorklundH2013,BjorklundDH2015,BjorklundKK2017,Bjorklund2018,BjorklundW2019}. The idea is to define a multivariate polynomial $P$ over a ring, along with  an efficient algorithm for its evaluation, with the property that $P$  is non-zero only if the graph has a Hamiltonian cycle. That polynomial can then be used to detect Hamiltonicity, by testing if $P$ is identically zero by evaluating $P$ at a random point using the efficient algorithm (Polynomial identity testing, PIT). The choice of ring is important and a somewhat delicate matter.
The basic observation is that using a larger ring increases the chance of making $P$ non-zero on many points, whereas a smaller ring typically makes it easier to come up with an efficient evaluation algorithm. We will use a large ring for the polynomial, but our sample space will only take values from a small subring on a large subset of the variables.
Our algorithm for evaluating $P$ follows a construction by~\cite{BjorklundKK2017} based on an exponential sum of weighted Laplacians of the graph, that in itself already describes a $2^n\operatorname{poly}(n)$ time algorithm. To get a running time below that, we take measures in designing our sample space so that many summands will be zero for a trivial reason, and we can find out which are not by solving a linear equation system over $\mathbb{Z}_2$. This is inspired by the algorithm in~\cite{BjorklundH2013} that lists solutions to a quadratic equation system over $\mathbb{Z}_2$ to sieve for the contributing terms. We list the solutions to a linear equation system by generating one solution from a Gaussian elimination followed by taking linear combinations of that solution with the null space (also found by the Gaussian elimination). This way we can list a superset of the summands that are non-zero and compute the sum to obtain the value of $P$ at our random point.

\section{The Algorithm}
\subsection{The Hamiltonicity Polynomial}
We begin by describing the polynomial $P$ we will be using.
Following~\cite{BjorklundW2019}, we will work on a slightly modified version $G=(V,A)$ of the $n$-vertex input graph $G_{\mbox{in}}=(V_{\mbox{in}},A_{\mbox{in}})$. We pick an arbitrary vertex $u\in V_{\mbox{in}}$, and replace $u$ with two new vertices $s$ and $t$, where $s$ retains all outgoing arcs from $u$, and $t$ retains all incoming arcs to $u$. Note that the Hamiltonian paths from $s$ to $t$ in this modified $G$ are in one-to-one correspondence with the Hamiltonian cycles in the original graph $G_{\mbox{in}}$, and that the average degree is not increased. In the following, we consider the problem of detecting a $s$-$t$ Hamiltonian path on the modified $n+1$ vertex graph $G$.

Fix a (commutative) ring $R$, and introduce a variable $z_{uv}\in R$ for each arc $uv\in A$. Let $\mathcal{H}(G)$ be the set of Hamiltonian $s$-$t$ paths in $G$, and consider the \emph{Hamiltonicity polynomial} $P_G$ as
\begin{equation}
P_G(z)=\sum_{H \in \mathcal{H}(G)} \prod_{uv\in H} z_{uv}.
\end{equation}

Our algorithm is based on an efficient way of evaluating $P_G(z)$ in a carefully chosen random point $z$ over a particular ring.
Note that we will write $z_{uv}$ in formulas to refer both to the formal variable and its value in $R$ according to a specific assignment $z:A\rightarrow R$.
In our analysis we will sometimes think of $P_G(z)$ as a formal polynomial in $z$ with coefficients from $R$, but in the algorithm itself, we always mean $P_G(z)$ to be an evaluation over $R$ of the polynomial $P_G$ in a specific point $z$.

For now, note that an evaluation of $P_G(z)$ in a random point $z$ can potentially be used as a fingerprint of existence of a Hamiltonian cycle in the original input graph: On one hand, the polynomial always evaluates to zero if the input graph has no Hamiltonian cycles (thus there are no false positives). On the other hand, by evaluating it in a random point, we will obtain a non-zero result if there is a Hamiltonian cycle in the input graph $G_{\mbox{in}}$, unless we are unlucky and the monomials happen to cancel each other. As mentioned above, there are two conflicting aspects to consider for a successful fingerprint design:
\begin{enumerate}
\item We want the ring and the sample space to be large enough so we can argue that the result is non-zero w.h.p. if the graph $G_{\mbox{in}}$ is Hamiltonian.
\item We want the ring and the sample space to have some structure that we can use to derive an efficient evaluation algorithm.
\end{enumerate}
The rest of our paper describes one way of balancing these aspects without having to resort to exponential size tabulation to enable a fast evaluation algorithm, by combining the graph Laplacian machinery in Bj\"orklund, Kaski, and Koutis~\cite{BjorklundKK2017} with the linear equation system modulo two listing idea from Bj\"orklund and Husfeldt~\cite{BjorklundH2013}. In Section~\ref{sec: samp}
 we will address the first aspect of balancing the fingerprint which is the major novel part. In Sections~\ref{sec: lap} and ~\ref{sec: perturb}
 we describe the basis of the algorithm in~\cite{BjorklundKK2017} (and subsequently in~\cite{BjorklundW2019}) that we will use, and
 in ~\ref{sec: speedup} and~\ref{sec: list}  we will address the second aspect of balancing the fingerprint by describing a linear algebra enumeration algorithm inspired by the algorithm in~\cite{BjorklundH2013}. Finally, in Section~\ref{sec: sum} we put the parts together into an algorithm for Theorem~\ref{thm:main}.
 We begin by defining the ring.
  
\subsection{The Choice of Ring}
\label{sec: ring}
We will work over the polynomial ring $R=\mathbb{Z}_{2^k}[x]/(x^m)$, i.e., polynomials in one variable truncated at degree $m$ with integer coefficients counted modulo $2^k$.
With foresight, both $k=k_R$ and $m=m_R$ will be $\operatorname{poly}(n)$, and hence an element in $R$ is described by $\operatorname{poly}(n)$ bits, and the arithmetic operations of addition and multiplication can both be done in $\operatorname{poly}(n)$ time. To compute a determinant of a matrix in $R^{n\times n}$, as we will need later, we may use Kaltofen's division-free algorithm~\cite{Kaltofen1992}, that uses $O(n^{3.5}\log n\log\log n)$ ring operations.
Altogether, the computation of the determinant of  an $n\times n$ matrix over the ring $R$ is a $\operatorname{poly}(n)$ time task.

\subsection{The Sample Space}
\label{sec: samp}
In this section, we describe the sample space over which we choose our point $z$ for polynomial identity testing. We will also argue that a randomly chosen point from the sample space has $P_G(z)\neq 0$ with non-zero constant probability when the graph $G$ has a Hamiltonian $s$-$t$ path. 
As we primarily are interested in the asymptotic form of the running time scaling, we will set the parameters somewhat arbitrarily for ease of calculations.
The sample space is parameterised by two positive integers $\tau$ and $\ell$ to be defined later in our analysis. The process to choose the point $z$ is given below:

\vspace{5mm}

\noindent \textbf{SamplePoint}(Returns $T$ and an assignment $z$ to be used for PIT of $P_G$)
\begin{enumerate}
\item Sample a subset $T\subseteq V\setminus \{s\}$ of size $\tau$ uniformly at random.
\item For every arc $uv,v\in T$, set $z_{uv}=1$.
\item For every arc $uv,v\not \in T$, set $z_{uv}=x^{w(uv)}$, where $w(uv)\in \{1,\cdots, \ell\}$ is a uniformly and independently randomly chosen integer.
\end{enumerate}

We next turn to proving that the choice of $z$ is good for PIT of $P_G$.
We will first look at the Hamiltonicity polynomial after assigning $z_{uv}=1$ for all $v\in T$, but for now still treat all other $z$-variables unassigned (left as formal variables).
We call these remaining variables $\tilde{z}$ with $\tilde{z}_{uv}=z_{uv}$ and consider the associated $T$-truncated polynomial $P_{G,T}(\tilde{z})$ obtained from $P_G$ after the variable substitution.
Define $\mathcal{H}_{T}(G)$ as the arc subsets of Hamiltonian $s$-$t$ paths in $\mathcal{H}(G)$ after the removal of any arc ending in $T$, i.e.,
\[
\mathcal{H}_{T}(G)=\{\cup_{uv\in H,v\not \in T} uv:H\in \mathcal{H}(G)\}.
\]
We call these the $T$-truncated Hamiltonian paths. We can write
\[
P_{G,T}(\tilde{z})=\sum_{H' \in \mathcal{H}_{T}(G)} e_{H'}\cdot\prod_{uv\in H'} \tilde{z}_{uv},
\]
where $e_{H'}$ counts the number of Hamiltonian cycles in $G$ with $T$-truncation $H'$.

We first need to prove that $P_{G,T}(\tilde{z})$ is not the zero-polynomial with high enough probability, when $G_{\mbox{in}}$ has a Hamiltonian cycle. Note that it may equal the zero-polynomial even in the presence of Hamiltonian cycles in $G_{\mbox{in}}$, when $e_{H'}$
is a multiple of $2^k$ for all $H'\in \mathcal{H}_{T}(G)$, as this would result in an annihilation in our ring $R$, where $k=k_R$ is the ring parameter in Section~\ref{sec: ring}. We will first argue that this doesn't happen with too large a probability.% for any input graph $G_{\mbox{in}}$ containing  a Hamiltonian cycle.

To prove this will not happen with some non-zero constant probability, let $H$ be \emph{any} fixed Hamiltonian $s$-$t$ path in $G$. In particular, our arbitrary choice of $H$ is independent of $T$. Let $H^T\in \mathcal{H}_{T}(G)$ be the $T$-truncation of $H$.
We will upper bound the expectation of $e_{H^T}$, the number of Hamiltonian $s$-$t$ paths in $G$ whose $T$-truncation matches $H^T$. For every $T$, define $S=S(T,H)\subseteq V$ to be the set of in-neighbors of $T$ along $H$, i.e., 
\begin{equation}
S=\{u:v\in T,uv\in H\}.
\end{equation}
Note that $S$ and $T$ are not necessarily disjoint, but of the same size $\tau$.
We consider the following bipartite graph $B_{H}$ obtained from an induced subgraph of $G$ as follows. $B_{H}$ has two parts, one representing the vertices in $S$, and one representing the vertices in $T$. All arcs in $B_{H}$ connects a vertex in the first part representing $S$ to a vertex in the second part representing $T$.
There is an arc from a vertex $u$ in the first part to a vertex $v$ in the second part, iff $uv$ is an arc in the induced subgraph $G[S\cup T]$. 
The following lemma tells us that the graph $B_{H}$ in expectation is not too dense.

\begin{Lem}
The expected number of arcs in $B_{H}$ is no more than
\[
{|S|}\left(1+|T|\frac{\delta}{n}\right).
\]
\end{Lem}
\begin{proof}
Let $d_v$ denote the outdegree of vertex $v\in V$. Consider a vertex $u\in S$. The arc from $u$ to the next vertex on $H$ is always present in $B_{H}$. The number of other arcs though, is in expectation $(d_u-1)\frac{|T|-1}{n}$ since the other vertices on $T$ apart from $u$'s out-neighbor on $H$ are uniformly distributed.
Hence, by the linearity of expectation, using that each vertex in $V\setminus\{t\}$ is included in $S$ with probability $|S|/n$, the expected number of arcs in $B_{H}$ is
\[
\sum_{u\in V\setminus\{t\}} \frac{|S|}{n}\left(1+(d_u-1)\frac{|T|-1}{n}\right)\leq|S|\left(1+|T|\frac{\delta}{n}\right).
\]
\end{proof}

This means, that if we choose the fixed size $\tau=|T|=|S|=\frac{n}{c\delta}$, we get expected average outdegree from the vertices in the part representing $S$ in $B_{H}$ bounded by $1+c^{-1}$. 
By Markov's inequality for a non-negative random variable $X$, 
\[
\Pr[X\geq \lambda\mathbb{E}[X]]\leq \frac{1}{\lambda},
\]
we can bound the probability that the average degree is not much larger:
\begin{Cor}
\label{cor: sparse}
The probability that the average outdegree of a vertex in $S$ in $B_{H}$ is at most 
\[
\left(1+\frac{1}{49}\right)\left(1+\frac{1}{c}\right),
\]
is at least $1/50$. 
\end{Cor}

We next observe that all Hamiltonian $s$-$t$ paths whose $T$-truncation is $H^T$ defines the same set $S$. This also means that every Hamiltonian $s$-$t$ path in $G$ whose $T$-truncation is $H^T$ must use some arc for each vertex in $S$ in $B_{H}$. 
The product of the outdegrees of vertices in $S$ is an upper bound on their number $e_{H^T}$.
 Hence, by the above corollary with probability at least $1/50$ there will be at most $((1+1/49)(1+c^{-1}))^\tau$ of them by the arithmetic mean-geometric mean inequality. By setting $k_R$ in Section~\ref{sec: ring} large enough so that
\[
2^{k_R}>\left(\left(1+\frac{1}{49}\right)\left(1+\frac{1}{c}\right)\right)^\tau,
\]
we will get a monomial with non-zero coefficient in $P_{G,T}(\tilde{z})$ with probability at least $1/50$. We note that it suffices to set
\[
k_R>\tau\log_2\left(\left(1+\frac{1}{49}\right)\left(1+\frac{1}{c}\right)\right)=\frac{n}{c\delta}\log_2\left(\left(1+\frac{1}{49}\right)\left(1+\frac{1}{c}\right)\right).
\]
We will need $k_R$ to be much smaller than $\tau$ in the evaluation algorithm described in the next sections in order to evaluate $P_G(z)$ fast. 
Setting $c=20$, say, we can thus use $k_R=\frac{\tau}{10}$ and conclude that $e_{H^T}<2^{k_R}$ with large enough probability, and hence that $P_{G,T}(\tilde{z})$ has at least one monomial. To summerise, we have that
\begin{Lem}
\label{lem: bound1}
With $\tau=\frac{n}{20\delta}$ and $k_R=\frac{n}{200\delta}$ (i.e., the parameters set as above), \emph{\textbf{SamplePoint}} returns $T$ so that
the formal polynomial
\[
P_{G,T}(\tilde{z})\neq 0,
\]
when the input graph has a Hamiltonian cycle, with probability at least $1/50$.
\end{Lem}

We next turn to arguing that $P_G(z)\neq 0$ (over the ring $R$) with high enough probability for the assignment $z$ returned by \textbf{SamplePoint}.
%To do so, we first recognise that we can write
%\begin{equation}
%\label{eq: 2i}
%P_{G,T}(\tilde{z})=\sum_{i=0}^{k_R-1} 2^iP_{G,T}^{(i)}(\tilde{z}),
%\end{equation}
%where $P_{G,T}^{(i)}(\tilde{z})$ consists of the sum of the monomials in $P_{G,T}(\tilde{z})$ whose coefficient is divisible by $2^i$ but not by $2^{i+1}$, after factoring out $2^i$ (meaning that all non-zero coefficients of monomials in $P_{G,T}^{(i)}(\tilde{z})$ are necessarily odd). From Lemma~\ref{lem: bound1}, we know that with some constant non-zero probability, there exists an $P_{G,T}^{(i)}(\tilde{z})$ that is not the zero polynomial.
%Let $l$ be the smallest integer such that $P_{G,T}^{(l)}(\tilde{z})$ is non-zero.
The next famous lemma by Mulmuley, Vazirani, and Vazirani~\cite{MulmuleyVV1987} shows that with $\ell$ large enough, we will be able to isolate a $T$-truncated Hamiltonian path in $\mathcal{H}_T(G)$ that is represented by a monomial in $P_{G,T}(\tilde{z})$.

\begin{Lem}[Isolation Lemma, Mulmuley, Vazirani, and Vazirani~\cite{MulmuleyVV1987}]
Let $m<M$ be two positive integers and let $\mathcal F$ be a nonempty family of subsets of $\{1,\cdots,m\}$. Suppose each element $x\in \{1,\cdots,m\}$
receives a weight $w(x)\in \{1,\cdots, M\}$ independently and uniformly at random. Define the weight of a set $S$ in $\mathcal F$ as $w(S)=\sum_{x\in S} w(x)$. Then, with
probability at least $1-\frac{m}{M}$, there is a unique set in $\mathcal F$ of minimum weight.
\end{Lem}

We apply the above Lemma, on the family $\mathcal F$ equal to $\mathcal{H}_T(G)$ \emph{further restricted} to those $T$-truncated Hamiltonian paths that are represented by a monomial in $P_{G,T}(\tilde{z})$, i.e., those $H'$ that have $2^{k_R} \nmid e_{H'}$. We use the weights $w(uv)$ set as in step $3$  of \textbf{SamplePoint} above to obtain $z$, with $\ell=100|A|$. We have with probability at least $1-1/100$ that a monomial exists with some unique weight $\mu$. In particular, with high enough probability, there is a Hamiltonian $s$-$t$ path whose $T$-truncation $H'$ is in $\mathcal{F}$ that will be isolated and contribute the value $e_{H'}x^\mu$ to $P_G(z)$. 
%This monomial cannot be cancelled in $R$ by other terms, as seen by inspecting~(\ref{eq: 2i}), since $P_{G,T}^{(i)}(\tilde{z})=0$ for $i<l$ by definition, and all other terms that may result in $x^\mu$ will have a coefficient divisible by $2^{l+1}$.

By setting the ring parameter $m_R>n\ell$ in Section~\ref{sec: ring}, we observe that this monomial in the polynomial ring is possible to detect.

\begin{Lem}
\label{lem: bound}
With $\tau=\frac{n}{20\delta}$,$k_R=\frac{n}{200\delta}$, $m_R>n\ell$, and $\ell=100|A|$ (i.e., all the parameters set as above), \emph{\textbf{SamplePoint}} returns $z$ so that
the polynomial
\[
P_{G}(z)\neq 0,
\]
when the input graph has a Hamiltonian cycle, with probability at least $1/100$.
\end{Lem}
The probability bound comes from the probability of $T$ being a good choice in Lemma~\ref{lem: bound1} to guarantee that there exists a $H'$ with $e_{H'}<2^{k_R}$ (1/50) after subtracting the probability that the Isolation Lemma was not successful in its isolation (1/100).

%Investigating 
%One interesting aspect of our algorithm besides the result itself, is the constraints it puts on (hypothetical) graphs that seems to be obstructions for an even more efficient algorithm.
%As we will see, the difficult graphs are those with a lot of Hamiltonian cycles. To be more precise, those graphs that if you cut a particular  Hamiltonian cycle open at a random large vertex subset $T$ of size $\tau=n/\delta$, by removing the arcs ending at a vertex in $T$, then almost surely, the number of ways to again assemble a Hamiltonian cycle from the pieces (including the original one) is divisible by $2^k$ for $k\in \Omega(\tau)$.

\subsection{The Laplacian}
\label{sec: lap}
Bj\"orklund, Kaski, and Koutis~\cite{BjorklundKK2017} observed that the number of Hamiltonian cycles in a directed graph can be evaluated as an inclusion--exclusion summation over a determinant of a polynomial matrix representing the graph. We will use their construction, not over the integers, but over the particular ring $R$ defined in Section~\ref{sec: ring}. This means we will lose the ability to count the Hamiltonian cycles, but it will also enable a faster evaluation of the inclusion--exclusion formula as we will demonstrate. We reiterate their construction here for the sake of completeness and easy reference.

The weighted Laplacian of the graph $G$, is a $(n+1)\times(n+1)$ polynomial matrix $L=L_G(y,z)$ with rows and columns indexed by the vertices $V$, in the variables $y_v$ for $v\in V\setminus\{t\}$, and variables $z_{uv}$ for $uv\in A$:
\begin{equation}
\label{eq: lap}
L_{i,j}=\left\{\begin{array}{ll} 
\sum_{wv\in A} z_{wv}y_w & \text{if } i=j=v \\ 
-z_{uv}y_u & \mbox{if } i=u,j=v,uv\in A\\ 0 & \mbox{otherwise}.\end{array}\right.
\end{equation}
The Laplacian \emph{punctured at the start vertex $s$}, is the matrix $L_{s}$ obtained by removing row and column $s$ from $L$. 
In~\cite{BjorklundKK2017}(their Theorem 5), it was observed that Tutte's directed version of the Matrix-Tree theorem of Kirchhoff~\cite{Tutte1948}, where $\operatorname{det}(L_s)$ is a polynomial in which each term corresponds to a directed spanning out-branching rooted at $s$, could be used to compute the Hamiltonicity polynomial. By the principle of inclusion--exclusion, letting $|y|$ denote the number of vertices $v$ for which $y_v=1$, we have

\begin{Lem}[Paraphrasing Equation (7) in Bj\"orklund, Kaski, and Koutis~\cite{BjorklundKK2017}]
\label{lem: hp}
\begin{equation}
\label{eq: hp}
P_G(z)=\sum_{y:(V\setminus\{t\})\rightarrow\{0,1\}}(-1)^{n-|y|}\operatorname{det}\left(L_s(y,z)\right).
\end{equation}
\end{Lem}
The summation is over all $2^{n-1}$ assignments $y:V\setminus\{t\}\rightarrow \{0,1\}$. 
Hence, with the formula in Lemma~(\ref{lem: hp}), we now have a way to evaluate $P_G(z)$ in a particular point $z$ in $2^n\operatorname{poly}(n)$ time.
We will next  see how we can speed-up the evaluation for a $z$ from our sample space.

\subsection{Random perturbations at $T$}
\label{sec: perturb}
Following~\cite{BjorklundKK2017} and~\cite{BjorklundW2019},
 we perturb the Laplacian matrices, without affecting the determinant, so that in expectation many summands in the above formula Eq.~\ref{eq: hp} are zeroed-out. We introduce new random variables $q_v\in \{0,1\}$ for $v\in T$, sampled uniformly and independently, where $T$ is the sampled set from \textbf{SamplePoint} and define the \emph{$q$-perturbed Laplacian} of $G$ as
\begin{equation}
L^q_{i,j}=
\left\{\begin{array}{ll} 
\sum_{wv\in A} z_{wv}y_w & \text{if } i=j=v,v\not \in T \\
\sum_{wv\in A} z_{wv}y_w-q_v & \text{if } i=j=v,v\in T \\
%q_v & \text{if } i=t,j=v\\
-z_{uv}y_u & \text{if } i=u,j=v, uv\in A\\
0 & \mbox{otherwise}.\end{array}\right.
\end{equation}
Comparing this to Eq.~\ref{eq: lap}, we have only added a term on some of the diagonal entries in the rows indexed by our sampled set $T$.
%The extra $q_v$ variables may be thought of as random arcs originating from $t$ onto $T$. 
Note that these extra $q_v$ variables do not affect the final inclusion--exclusion sum, as only the monomials representing Hamiltonian paths from $s$ to $t$ are counted, in particular only monomials with all $n$ $y_u$-variables for $u\in V\setminus\{t\}$, confer~\cite{BjorklundKK2017} for a proof. Hence, irrespective of $q$, we can still compute the Hamiltonicity polynomial as:
\begin{equation}
\label{eq: ham}
P_G(z)=\sum_{y:(V\setminus\{t\})\rightarrow\{0,1\}} (-1)^{n-|y|}\operatorname{det}\left(L^q_s(y,z)\right).
\end{equation}
What we have gained by doing this, is that the probability that a row indexed by $i\in T$ has its diagonal entry divisible by two, is $1/2$, independently of other rows. We will next see how we can use this.

\subsection{Efficient Evaluation of $P_G(z)$ given $T$}
\label{sec: speedup}

The basic idea is the same underlying the speed-ups in~\cite{BjorklundH2013,BjorklundKK2017,BjorklundW2019}. We make sure that in expectation, many summands in Eq.~(\ref{eq: ham}) will be trivially zero. Then, to evaluate the formula it suffices to list only the summands that are not trivially zero, so-called \emph{contributing} terms, and sum up their contributions. Here, with ``trivially zero'', we will mean matrices that has at least $k=k_R$ rows of the matrix among the rows indexed by a vertex in the sampled set $T$ from \textbf{SamplePoint} with all coefficients even. 
To see that such a term is zero, we merely have to recall Leibniz's determinant expansion of a matrix $M=\{m_{i,j}\}$:
\[
\operatorname{det}(M)=\sum_{\sigma\in S_n} \operatorname{sgn}(\sigma)\prod_{i=1}^n m_{i,\sigma(i)},
\]
where $S_n$ is the set of all permutations on $n$ elements. Note in particular that in every term there is one element from each row. Hence, if the matrix has $k$ rows in which every monomial $ax^b$ in a ring element has $a$ even, the product (over $\mathbb{Z}$) must be divisible by $2^k$ and hence cancel in the ring $R$. 

Our algorithm to compute $P_G(z)$ will list the terms in Eq.~\ref{eq: ham} that have at least one odd coefficient in some ring element in at least $\tau-k+1$ of the rows of $L_s$ representing vertices in $T$. This is what is required to be a contributing terms.
The algorithm outline is postponed to the next section.

We begin by arguing that, in expectation over the random $q$ values, there are not too many contributing terms.
Recalling Eq.~\ref{eq: lap}, and inspecting any such row in the matrix $L_s(y)$ for a vertex $v\in T$, we see that 

\begin{enumerate}
\item Off-diagonal entries are zero if $y_v=0$ or $v=t$, 
\item The diagonal entry is divisible by two if 
\[
\sum_{wv\in A} y_w=q_v (\mbox{ mod }2),
\]
remembering that $z_{wv}=1$ for all $wv\in A$ with $v\in T$.
\end{enumerate}
\vspace{5mm}

Fix an assignment $y:V\setminus\{t\}\rightarrow \{0,1\}$, and let $Z_y\subseteq T$ be the vertices $u$ for which the assignment sets $y_u=0$, along with $t$ if $t\in T$.
From the above, the probability over the random $q$ values, of the event $\varepsilon_y$ that a fixed assignment $y$ does \emph{not} result in a trivially zero term in Eq.~\ref{eq: ham}, is 
\begin{equation}
\label{eq: prob}
\Pr_q[\varepsilon_y]= \left(\frac{1}{2}\right)^{|Z_y|}\sum_{i=0}^{k-1} \binom{|Z_y|}{i}.
\end{equation}
Here we use that the diagonal entry of a row indexed by a vertex in $Z_y$ is even with probability $1/2$ independently of other rows as argued in Section~\ref{sec: perturb}.
Let $Y$ be the random variable equal to the number of assignments that are contributing. Then, in expectation
\begin{equation}
\label{eq: exp}
\mathbb{E}[Y]=\sum_{y\in V\setminus\{t\}} \Pr_q[\varepsilon_y].
\end{equation}
We can bound the expectation as
\begin{Lem}
\label{lem: runtime}
\[
\mathbb{E}(Y)\in 2^{n-\Omega\left(\frac{n}{\delta}\right)}.
\]
\end{Lem}
\begin{proof}
From Eq.~\ref{eq: prob} and Eq.~\ref{eq: exp} we have
\begin{equation}
\label{eq: split}
\mathbb{E}[Y] \leq \sum_{\substack{y\in V\setminus\{t\}\\ |Z_y|<\frac{\tau}{3}}} 1 +\sum_{\substack{y\in V\setminus\{t\}\\ |Z_y|\geq \frac{\tau}{3}}}  \left(\frac{1}{2}\right)^{|Z_y|}\sum_{i=0}^{k-1} \binom{|Z_y|}{i}.
\end{equation}

The left term in Eq.~\ref{eq: split} is
\[
2^{n-\tau}\left(\sum_{i=0}^{\tau/3-1} \binom{\tau}{i}\right) \in 2^{n-\Omega(\tau)},
\]
and the right term in Eq.~\ref{eq: split} is less than
\[
2^{n}\left(\frac{1}{2^{\tau/3}}\sum_{i=0}^{k-1} \binom{\tau/3}{i}\right) \in 2^{n-\Omega(\tau)},
\]
after remembering $k<\tau/10$ and $\tau=\frac{n}{20\delta}$, and noting that 
\[
\frac{1}{2^{\gamma}}\sum_{i=0}^{k-1} \binom{\gamma}{i},
\]
for $\gamma \in\{\tau/3,\cdots \tau\}$ is maximised for $\gamma=\tau/3$.
The stated bound in the Lemma follows.
\end{proof}
\subsection{Listing Contributing Terms}
\label{sec: list}
We finally describe how to list the contributing term assignments $y:V\setminus\{t\}\rightarrow \{0,1\}$ needed to compute $P_G(z)$ via Eq.~\ref{eq: ham} for a fixed $q$.
The idea is to test for each partial assignment $y^*:T\setminus \{t\}\rightarrow\{0,1\}$ with the interpretation that $y_v=y^*_v$ for $v\in T\setminus \{t\}$, and each way of assigning parities $p:Z_{y^*}\rightarrow \{0,1\}$ to the diagonal entries of the vertices in $Z_y$ that is consistent with a not trivially zero assignment, i.e., $p$ takes the value $0$ on at most $k-1$ rows. We then notice that these diagonal entries in $Z_y$ describe a linear equation system over the variables in $y$ outside of $T$. The equation system $E(y^*,p)$ consists of the equations (modulo two)
\[
\sum_{wv\in A} y_{w}+q_v=p_v,
\]
for each $v\in Z_y$, where we replace each variable $y_{w}$ with $w\in T$ for its value $y^*_{w}$. We can list all solutions to this equation system by Gaussian elimination. We first solve for one solution and a null space basis. We next can enumerate all solutions by taking all linear combinations of the null space basis vectors with the solution. In summary our streaming procedure that generates all contributing terms' assignments is (we will think of it as a background process generating the solutions one-by-one):
\vspace{5mm}

\noindent \textbf{ListingTerms}(outputs contributing assignments $y:V\setminus\{t\}\rightarrow \{0,1\}$ needed for Eq.~\ref{eq: ham})
\begin{enumerate}
\item For each $y^*:T\setminus\{t\} \rightarrow \{0,1\}$,
\item \hspace{5mm} For each $p:Z_{y^*}\rightarrow \{0,1\}$ with $|p|>|Z_{y^*}|-k$,
\item \hspace{10mm} Report every solution $y$ to $E(y^*,p)$.
\end{enumerate}

Note that this lists every contributing term's assignment once since each $y$ has precisely one restriction $y^*$ on $T$ and matches one of the tested $p$'s. To bound the running time, we know from Lemma~\ref{lem: runtime} that the output number of $y$ assignments are at most $2^{n-\Omega(\frac{n}{\delta})}$ in expectation. This dominates the running time, since the number of equation systems considered, each of which can be solved in polynomial time in the sense of providing a parameterisation of the solution space as one solution vector along with the null space, is at most $3^\tau\in 2^{O(n/\delta)}$.
\subsection{High-Level Algorithm}
\label{sec: sum}
Putting the parts of the previous sections together, we are ready to give the high-level description of our algorithm in Theorem~\ref{thm:main} as
\vspace{5mm}

\noindent \textbf{DecideHamiltonicity}(answers whether input $G_{\mbox{in}}$ has a Hamiltonian cycle)
\begin{enumerate}
\item Repeat for $100\log n$ times:
\item \hspace{5mm} Call \textbf{SamplePoint} to obtain point $z$ and subset $T$.
\item \hspace{5mm} Pick a $q$ uniformly at random.
\item \hspace{5mm} While there are still contributing terms:
\item \hspace{10mm} Get next $y$ from background process \textbf{ListingTerms}.
\item \hspace{10mm} If the number of generated terms is too big, continue to next outer repetition.
\item \hspace{10mm} Add $y$'s contribution to $P_G(z)$.
\item \hspace{5mm} If $P_G(z)\neq 0$ break and output ``Yes''.
\item Output ``No''.
\end{enumerate}

In particular, there is no need to store the list of $y$ assignments explicitly, but rather we use them one by one as they are generated to update the sum in Eq.~\ref{eq: ham}. From Lemma~\ref{lem: bound} we know that a false negative happens with probability $1-1/100$.
Since we pick $100\log n$ sample points $z$, independently of each other, we will be unsuccessful in all of them with probability $(1-1/100)^{100\log n}<n^{-1}$. From Section~\ref{sec: list} we know the number of contributing term assignments and the running time of \textbf{ListingTerms} is $2^{n-\Omega(\frac{n}{\delta})}$ in expectation. If the number of generated terms are more than $n$ times the expected value, we abort this $z,T,q$-value repetition at step 6 of the algorithm. This also happens only with probability $n^{-1}$ by Markov's inequality. 
Altogether, the probability of a false negative is at most $\frac{2}{n}$.
This concludes the proof of Theorem~\ref{thm:main}. 

\section*{Acknowledgements}
We are very grateful to an anonymous reviewer who pointed out a serious flaw in an earlier proof attempt of Lemma~\ref{lem: bound}.

\bibliographystyle{abbrv}
\bibliography{paper}

\end{document}